\documentclass[english, a4paper]{lipics-v2021} 
\hideLIPIcs
\nolinenumbers
\usepackage[utf8]{inputenc}

\usepackage{microtype,hyperref,graphicx}
   \hypersetup{%
      breaklinks,%
      colorlinks=true,%
      urlcolor=[rgb]{0.25,0.0,0.0},%
      linkcolor=[rgb]{0.5,0.0,0.0},%
      citecolor=[rgb]{0,0.2,0.445},%
      filecolor=[rgb]{0,0,0.4},
      anchorcolor=[rgb]={0.0,0.1,0.2}%
   }
   
\title{Computing the Girth of a Segment Intersection Graph}

\author{Timothy M. Chan}{Siebel School of Computing and Data Science, University of Illinois Urbana-Champaign, USA}{tmc@illinois.edu}{https://orcid.org/0000-0002-8093-0675}{Work supported by NSF Grant CCF-2224271.} %
\author{Yuancheng Yu}{Siebel School of Computing and Data Science, University of Illinois Urbana-Champaign, USA}{yyu51@illinois.edu}{https://orcid.org/0009-0003-5570-1963}{}

\titlerunning{}
\authorrunning{T.\,M. Chan and Y. Yu}

\Copyright{Timothy M. Chan and Yuancheng Yu}



\renewcommand{\paragraph}[1]{\subparagraph*{#1}}

\newcommand{\eps}{\varepsilon}
\newcommand{\mathify}[1]{\ifmmode{#1}\else\mbox{$#1$}\fi}
\newcommand{\abs}[1]{\mathify{\left| #1 \right|}}

\newtheorem{fact}[theorem]{Fact}

\newcommand{\IGNORE}[1]{}
\newcommand{\R}{\mathbb{R}}
\newcommand{\OO}{\widetilde{O}}

\newenvironment{pproof}{\begin{proof}}{\end{proof}}

\newcommand{\sizeP}{p}

\newcommand{\sep}{S}
\newcommand{\sepq}{Q}
\newcommand{\sepp}{P}
\newcommand{\tree}{\mathcal{T}}

\begin{document}
\maketitle

\begin{abstract}
We present an algorithm that computes the girth of the intersection graph of $n$ given line segments in the plane in $O(n^{1.483})$ expected time. This is the first such algorithm with $O(n^{3/2-\varepsilon})$ running time for a positive constant $\varepsilon$, and makes progress towards an open question posed by Chan (SODA 2023). The main techniques include (i)~the usage of recent subcubic algorithms for bounded-difference min-plus matrix multiplication, and (ii)~an interesting variant of the planar graph separator theorem. The result extends to intersection graphs of connected algebraic curves or semialgebraic sets of constant description complexity.
\end{abstract}

\section{Introduction}

Intersection graphs of geometric objects have received considerable attention in computational geometry.  They
generalize planar graphs (every planar graph is an
intersection graph of disks in 2D by the Koebe--Andreev--Thurston theorem),
and they arise in numerous applications (for example, unit disk graphs can be used to model communication networks, and 
intersection graphs of line segments in 2D can be used to model road networks).
Over the years, researchers have studied many standard computational problems on geometric intersection graphs, aiming
for algorithms that run faster than those for general graphs.  The problems considered include
the computation of connected components~\cite{DBLP:journals/algorithmica/AgarwalK96}, all-pairs shortest paths~\cite{DBLP:journals/jocg/ChanS19}, diameter~\cite{DBLP:conf/compgeom/BringmannKKNP22,DBLP:journals/corr/abs-2510-16346/ChanCGKLZ25}, minimum cut~\cite{DBLP:journals/comgeo/CabelloM21}, maximum matching~\cite{DBLP:journals/dcg/BonnetCM23},
and finding small-sized subgraphs (such as triangles, cycles of a fixed length, cliques of a fixed size, etc.)~\cite{DBLP:conf/soda/Chan23}.

In this paper, we focus on another of the most basic graph algorithm problems, which has not been as well explored in the context of geometric intersection graphs: the computation of the \emph{girth}, i.e., the length of the shortest cycle (see Figure~\ref{fig:girth_of_segment_intersection_graph}).
Girth is related to various topics in graph theory, including the chromatic number, treewidth, diameter, Ramsey theory, etc.~\cite{DBLP:journals/dm/Bollobas78,cook1975chromatic,lovasz1968chromatic,erdos1959graph,DBLP:journals/jct/ChandranS05,DBLP:books/daglib/0030488/Diestel12}. It can be computed in $O(mn)$ or $O(n^{\omega})$ time for general
unweighted undirected graphs with $n$ vertices and $m$ edges \cite{DBLP:journals/siamcomp/ItaiR78}, where $\omega<2.371339$ is the matrix multiplication exponent~\cite{DBLP:conf/soda/AlmanDWXXZ25}. 
For the case of planar graphs, a series of papers~\cite{DBLP:conf/icalp/Djidjev00,DBLP:journals/talg/Djidjev10,DBLP:conf/soda/ChalermsookFN04,DBLP:journals/siamdm/WeimannY10} culminated in a linear-time algorithm by Chang and Lu~\cite{DBLP:journals/siamcomp/ChangL13}.
However, for intersection graphs,
we are aware of only two main algorithmic results:

\begin{figure}
\centering
\includegraphics[scale=.4]{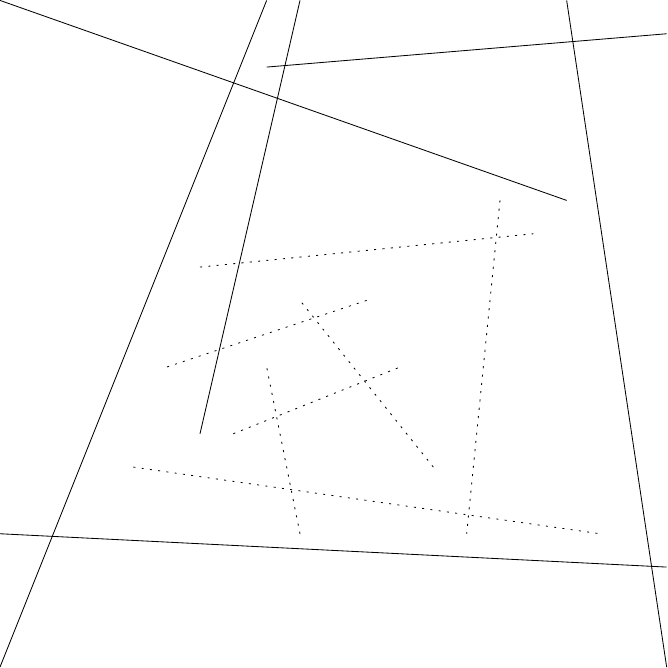}
 \caption{An example of the shortest cycle of a segment intersection graph (segments not in the shortest cycle are dotted). The girth is $6$.}\label{fig:girth_of_segment_intersection_graph}
\end{figure}

\begin{enumerate}
\item Kaplan et al. at ESA'19~\cite{DBLP:conf/esa/KaplanKMRSS19} gave an $O(n\log n)$-time algorithm for computing the girth of the intersection graph
of $n$ disks in 2D (in the case of unit disks, this was subsequently improved by Klost and Mulzer~\cite{DBLP:journals/corr/abs-2405-01180/KlostM24} to an $O(n)$-time algorithm which
does not require the disks be given, just the graph).

\smallskip
\item Chan at SODA'23~\cite{DBLP:conf/soda/Chan23} gave an $\OO(n^{3/2})$-time\footnote{
The $\OO$ notation hides polylogarithmic factors.
} algorithm for computing the girth of the intersection graph
of $n$ line segments in 2D.
\end{enumerate}

\noindent
Chan's work naturally raises the following question:
\begin{quote}
\begin{center}
  \em Is there an $O(n^{3/2-\eps})$-time algorithm for computing the girth of the intersection graph of $n$ line segments for some constant $\eps>0$?
\end{center}
\end{quote}

The case of line segments is particularly important, not only because of direct relevance to road networks, for example, but also because it serves as a stepping stone towards the general case of algebraic curves or semialgebraic sets in the plane.

\subsection{New result}
We answer the above open question in the affirmative by presenting the first algorithm that breaks the $n^{3/2}$ barrier. Specifically, our algorithm computes the girth for the intersection graph of $n$ line segments in 2D in $O(n^{1.483})$ expected time. The algorithm also extends to intersection graphs of algebraic curves or semialgebraic sets where each object is connected and has constant description complexity---essentially, most families of constant-size objects in 2D that are commonly encountered.

\subsection{Why is the problem challenging?}

First note that
Kaplan et al.'s approach~\cite{DBLP:conf/esa/KaplanKMRSS19} is not applicable for our problem, as it exploited special properties of disks that do not hold for line segments (namely, that any intersection graph of disks with girth more than 3 must be planar). For the related problem of finding short cycles, Chan~\cite{DBLP:conf/soda/Chan23}
extended Kaplan et al.'s result to arbitrary fat objects in any constant dimension, by using a shifted quadtree approach, but line segments are not fat.

For intersection graphs of line segments, Chan~\cite{DBLP:conf/soda/Chan23} has already given an $O(n^{1.408})$-time algorithm for detecting a triangle (i.e., 3-cycle), and
an $\OO(n)$-time algorithm for detecting a 4-cycle.  Thus, it suffices to handle the remaining case when the girth is greater than 4.  It is known that any intersection graph of $n$ line segments without 4-cycles must be sparse, i.e., has $O(n)$ edges. We can explicitly construct the arrangement of the line segments and form a planar graph of $O(n)$ size whose vertices are the endpoints and intersection points.  This way, our problem is reduced to a variant of the girth problem on a planar graph, where
the cost of a cycle is the number of turns---turning at an intersection point has unit cost, but going straight through an intersection has zero cost.

However, this variant does not seem directly reducible to the standard girth problem on planar graphs.  One could create a gadget at each intersection point and obtain a graph with $0$/$1$ edge weights, such that the weight of a path in the new graph corresponds to the number of turns along the path.  But the resulting graph would no longer be planar (it would be \emph{$1$-planar} instead); besides, unintended small cycles would be introduced in such a gadget. In addition, many properties about shortest paths in planar graphs no longer hold.  For example, two shortest paths in a planar graph
cannot cross more than once (if ties are broken carefully), but two minimum-turn paths may cross
numerous times.  Consequently, many techniques devised for shortest-path-related problems in planar graphs (e.g., based on some kind of Monge property~\cite{DBLP:journals/jcss/FakcharoenpholR06}, or abstract Voronoi diagrams~\cite{DBLP:journals/talg/Cabello19,DBLP:journals/siamcomp/GawrychowskiKMS21}, or VC-dimension arguments~\cite{DBLP:conf/stoc/LiP19,DBLP:journals/siamcomp/DucoffeHV22,DBLP:conf/soda/LeW24,DBLP:journals/corr/abs-2510-16346/ChanCGKLZ25})
seem to fall apart for our problem.

One technique that still works is \emph{separators}.
In fact, Chan's previous $\OO(n^{3/2})$-time algorithm is obtained by a simple application of separators.  We can divide the graph into two parts by a separator of size $O(\sqrt{n})$, and recurse on both parts.  This handles the case when the shortest cycle lies completely in one part. To deal with the remaining case when the shortest cycle passes through a separator vertex, we can run BFS from each of the $O(\sqrt{n})$ separator vertices, in $\OO(n^{3/2})$ total time.

For the case when the girth is a constant (or for detecting a cycle of a fixed constant length), Chan~\cite{DBLP:conf/soda/Chan23} has shown how to adapt the separator approach to obtain a faster algorithm running in $\OO(n^{\omega/2})\le O(n^{1.186})$ time. The main observation is that the $O(\sqrt{n})$ calls to BFS is overkill---we only need to compute shortest path distances with \emph{both} sources and destinations from the $O(\sqrt{n})$ separator vertices.  We can compute such distances recursively as we compute the girth.  Combining the output of the two recursive calls involves computing min-plus products of  $O(\sqrt{n})\times O(\sqrt{n})$ distance matrices.
When the matrix entries are integers bounded by a constant, min-plus product reduces to $O(1)$ number of standard (Boolean) matrix product.  This is roughly how the $\OO(n^{\omega/2})$ time bound was obtained.

For larger girth, however, this approach does not work well, because the complexity of min-plus product grows for larger integer values. (Another issue is the use of \emph{color-coding}~\cite{DBLP:journals/jacm/AlonYZ95} in Chan's cycle detection algorithm~\cite{DBLP:conf/soda/Chan23}, which led to exponential dependence on the cycle length.)

On the other hand, when the girth $g$ is very large, one could obtain an $\OO(n^2/g)$-time algorithm by a well-known hitting set
technique~\cite{DBLP:journals/siamcomp/UllmanY91,
DBLP:journals/jacm/Zwick02} (there exists a subset of $\OO(n/g)$ vertices which hit all shortest paths of length $\Omega(g)$).  The main difficulty lies in the intermediate case, when the girth is not too small and not too large.

\subsection{Our approach}

\subparagraph*{Ingredient 1: Bounded-difference min-plus matrix multiplication.}
The previous separator-based approach faces an $n^{3/2}$ barrier because min-plus matrix multiplication is widely conjectured to require near-cubic time (according to the ``All-Pairs Shortest Paths (APSP) hypothesis'' from fine-grained complexity~\cite{virgisurvey}).
Our first idea is to employ recent results which solve certain powerful special cases of min-plus product in truly subcubic time---namely, the case when the input matrices are integer-valued and satisfy a \emph{bounded-difference}
property, i.e., the difference between any two adjacent entries (in each row or in each column) is at most a constant.
Chan and Lewenstein~\cite{DBLP:conf/stoc/ChanL15} gave the first truly subquadratic algorithm for the related  min-plus convolution problem for bounded-difference sequences, while
Bringmann et al.~\cite{DBLP:journals/siamcomp/BringmannGSW19} gave the first truly subcubic algorithm for bounded-difference min-plus product.
Bringmann et al.'s original algorithm ran in $O(n^{2.8244})$ time, and after a series of improvements~\cite{DBLP:conf/icalp/GuPWX21,DBLP:conf/soda/ChiDX22}, Chi et al.~\cite{DBLP:conf/stoc/ChiDXZ22} obtained the current best running time of $\OO(n^{(3+\omega)/2})\le O(n^{2.686})$ (see Section~\ref{subsec:prelim:bdd}).

Bounded-difference min-plus convolution and matrix multiplication have recently been applied to solve a variety of problems from different domains, including language edit distance and RNA folding~\cite{DBLP:journals/siamcomp/BringmannGSW19}, single-source replacement paths and Range Mode~\cite{DBLP:conf/icalp/GuPWX21}, tree edit distance~\cite{DBLP:conf/focs/Mao21}, approximate APSP~\cite{DBLP:conf/soda/SahaY24}, knapsack~\cite{DBLP:conf/icalp/BringmannC22,DBLP:conf/esa/BringmannDP24}, etc.
Our algorithm, being one of the first applications to a computational geometry problem, is thus noteworthy.

\subparagraph{Ingredient 2: A variant of the planar-graph separator theorem.}
Why is the bounded-difference property relevant in our application?
Suppose the separator $S$ is connected, e.g., is a cycle (as in Miller's version of the separator theorem~\cite{DBLP:journals/jcss/Miller86}).
Then for any vertex $q$ and two adjacent separator vertices $s,s'\in S$,
$d(s,q)$ and $d(s',q)$ differ by at most 1 due to the triangle inequality.
Thus, the distance matrices we encounter during recursion would indeed satisfy the bounded-difference property.

Unfortunately, cycle separators do not always exist when the given planar graph is not $2$-connected.
Our key idea is to use an interesting variant of the planar-graph separator theorem, which holds for any planar graphs:
given any $\alpha>0$, there always exists a separator $S$ consisting of 
two subsets $P$ and $Q$, where 

\begin{itemize}
    \item the first subset $P$ has size $O(n^{1/2+\alpha})$ but has at most 2 connected components (in fact, at most 2 paths);
    \smallskip
    \item the second subset $Q$ may not be connected but has size $O(n^{1/2-\alpha})$.
\end{itemize}

We are not aware of this variant of the separator theorem explicitly stated before (though we will not be surprised if it appeared before), but it follows by straightforward modification of standard proofs.
See Section~\ref{sec:sep} for the precise statement and proof.

This variant of separators is exactly the tool we need to break the $n^{3/2}$ barrier: 
since $P$ has only $O(1)$ connected pieces, we can use bounded-difference min-plus products to compute distances involving $P$ faster; on the other hand, since $Q$ has size $\ll\sqrt{n}$, we can afford to compute distances involving $Q$ more naively.
We choose the trade-off parameter $\alpha$ carefully to balance cost.
This is the main high-level idea of our algorithm.  Quite a bit of technical effort, however, is needed to make the idea work 
(for example, the number of connected components may grow as we recurse, but not by much, as we will show, and we eventually need two separate recursive algorithms, one for computing distances and one for computing the girth).
See Sections~\ref{sec:dist}--\ref{sec:girth} for the detailed algorithm and analysis.

\subparagraph{Ingredient 3: Ensuring simplicity of walks.}
When computing a shortest path between two vertices,
we do not need to worry about ensuring simplicity of the path, since a shortest walk must automatically be a shortest path (if a vertex is repeated, we can take a short-cut to obtain a shorter walk).
However, more care is needed for the shortest cycle problem in undirected graphs.  When concatenating several shortest paths to form a candidate cycle (via min-plus products of several distance matrices), one potential issue is that the resulting closed walk might become trivial (we might be completely retracing our steps, so that short-cutting yields an empty walk!).
The earlier algorithm by Chan~\cite{DBLP:conf/soda/Chan23} for detecting constant-length cycles used the color-coding technique~\cite{DBLP:journals/jacm/AlonYZ95} to ensure simplicity, but unfortunately color-coding yields exponential dependence on the cycle length.
For general graphs, there are other strategies to ensure simplicity,
(e.g., by computing shortest paths in different subgraphs with some vertices and edges removed,  or with directions added, so as to avoid concatenating two identical shortest paths).
But ensuring simplicity
is trickier in our setting, since we need to retain
the bounded-difference property if we work with new subgraphs.
Still, we have found one way to deal with this technical issue, using ideas given in Section~\ref{sec:simplicity}.

\section{Preliminaries}

Before presenting our algorithm,
we begin by providing more details on the three ingredients mentioned in the introduction.

Let $\delta_G(u,v)$ denote the shortest-path length between $u$ and $v$ in an unweighted undirected graph $G$.
Let $\delta_G[A,B]$ denote the distance matrix $\{\delta_G(u,v)\}_{u\in A, v\in B}$ for two vertex subsets $A$ and $B$ in a graph $G$.
Let $G[S]$ denote the subgraph of $G$ induced by a vertex subset $S$.

\subsection{Bounded-difference min-plus product}\label{subsec:prelim:bdd}

Given two matrices $A,B$, their \emph{min-plus product} $A\star B$ is defined as
\[
(A\star B)(i,j)=\min_k (A(i,k)+B(k,j)).
\]

As mentioned, we will use a known subcubic algorithm for computing the min-plus product of two $n\times n$ integer matrices that satisfy a bounded-difference property.
The following is due to Chi et al.~\cite{DBLP:conf/stoc/ChiDXZ22}:

\begin{lemma}\label{lem:bdd:diff}
Given two integer $n\times n$ matrices $A$ and $B$ where each column of $A$ satisfies the bounded-difference property (i.e., $|A(i,j)-A(i+1,j)|\le O(1)$ for every $i,j$), we can compute their min-plus product $A\star B$ in $\OO(n^{(3+\omega)/2})\le O(n^{2.686})$ expected time.

The same holds if instead  each row of $A$ satisfies the bounded-difference property (i.e., $|A(i,j)-A(i,j+1)|\le O(1)$), or if the property holds for $B$ rather than $A$.
\end{lemma}

\subsection{A variant of the planar separator theorem}\label{sec:sep}

The second ingredient is an interesting variant of the planar-graph separator theorem, in which the separator is divided into two subsets $P$ and $Q$: $P$ may be larger but has only 2 connected pieces, while $Q$ may not be connected but is smaller.  The precise statement is given below (the original theorem corresponds to the case when $p=\sqrt{n}$ with unit loads, while our application requires the case when $p\gg\sqrt{n}$).
The proof follows by modifying standard proofs
\cite{DBLP:journals/siamam/LiptonT79,KMbook,Jeffnotes}.

\begin{lemma}\label{lem-sep}
Let $G=(V,E)$ be a planar graph with $n$ vertices, 
where each vertex has a nonnegative \emph{weight} and a nonnegative \emph{load},
with total weight $W$, total load $\Lambda$, and the maximum individual vertex weight is at most $\eps W$ for a sufficiently small constant $\eps>0$. For a given parameter~$\sizeP$,
we can partition $V$ into subsets $V_1,V_2,\sep$, in $O(n)$ time, such that

\begin{itemize}
\item the total weight of $V_i$ is at most $3W/4$ for each $i\in\{1,2\}$;
\item there are no edges in $V_1\times V_2$;
\item $\sep$ can be further partitioned into $\sepp$ and $\sepq$
such that $\sepp$ is the union of at most two paths of length $O(p)$, and 
$\sepq$ has total load at most $O(\Lambda/p)$.
\end{itemize}
\end{lemma}

\begin{pproof}
We first assume that $G$ is connected.
Let $\tree$ be a BFS tree of $G$ from a fixed source vertex.
For any $uv\not\in\tree$,
let $\mathrm{cycle}(\tree,uv)$ denote the unique cycle in $\tree\cup\{uv\}$
(called the \emph{fundamental cycle}), which consists of $uv$ and the unique path in $\tree$ connecting $u$ and $v$.

Let $\Sigma$ be a (combinatorial) triangulation of $\tree$.
Define face weights of $\Sigma$ by examining each vertex $v$, picking an arbitrary incident face $f$, and adding $v$'s weight to $f$.
Then the total face weight is $W$.
Consider the dual tree with respect to $\tree$ in $\Sigma$ (which has maximum degree 3). 
By standard results on tree separators~\cite{KMbook,Jeffnotes},
we can find an edge which separates this dual tree into two parts each with total face weight at most $3W/4$.  In $\Sigma$, this corresponds to an edge $uv$ such that
either side of $C=\mathrm{cycle}(\tree,uv)$ has total face weight at most $3W/4$; this also bounds the total weight of the vertices strictly inside/outside $C$. If $\abs{C}\le p$ then the desired separator can obtained by setting $\sepp=C$ and $\sepq=\emptyset$. Below we may assume $\abs{C}>p$.

Let $x$ be the lowest common ancestor of $u,v$ in $\tree$, 
and $\tree'$ be the subtree of $\tree$.
Let $L_i\subset V$ be the set of vertices in level $i$ of $\tree'$ (with $L_0=\emptyset$). 
Let $m$ be the smallest integer such that the total weight of vertices in the first $m$ levels of $\tree'$ is at least $W/2$. 
By the pigeonhole principle,
there exists $\ell\in[m-p,m)$ such that 
$L_\ell$ has total load at most $\Lambda/p$.
Similarly there exists $h\in(m,m+p]$ such that 
$L_h$ has total load at most $\Lambda/p$.

If $V_{\le\ell}=\bigcup_{i=1}^{\ell} L_i$ has weight at least $W/4$, then since $V_{\le\ell}$ has weight at most $W/2$ by definition of $m$, the desired separator can be obtained by setting $\sepq=L_{\ell}$  and $\sepp=\emptyset$.

Similarly, if $V_{\ge h}=\bigcup_{i=h}^{\infty} L_i$ has weight at least $W/4$, then the desired separator can be obtained by setting $\sepq=L_h$  and $\sepp=\emptyset$.

In the remaining case, $V_m=\bigcup_{i=\ell+1}^{h-1} L_i$ has weight at least $W/2$. Let $V_{\mathrm{in}},V_{\mathrm{out}}$ be the set of vertices in $V_m$ that are strictly inside or outside $C$, respectively. Without loss of generality, say  $V_{\mathrm{in}}$ has weight at least $W/4$. Recall that  $V_{\mathrm{in}}$ has weight at most $3W/4$ by definition of $C$. Let $P_u$ (resp.\ $P_v$) be the portion of the shortest path from $x$ to $u$ (resp.\ $v$) with levels between $\ell$ and $h$. Then $P_u$ and $P_v$ are paths of length $O(p)$. The desired separator can be obtained by setting $\sepp=P_u\cup P_v$ and $\sepq=L_{\ell}\cup L_h$.
See Figure~\ref{fig:separator}.

Finally, we consider the general case when $G$ may have multiple connected components.
If all components have weight at most $W/2$, a trivial separator exists with $\sep=\emptyset$.  On the other hand, if $G$ has one component of weight at least $W/2$, we can apply the lemma to the largest-weight component and add the remaining components to the smaller of the two subsets $V_1,V_2$ in
terms of weight (which would still have total weight at most $3W/4$).
\end{pproof}

\begin{figure}
\centering
\includegraphics[scale=.5]{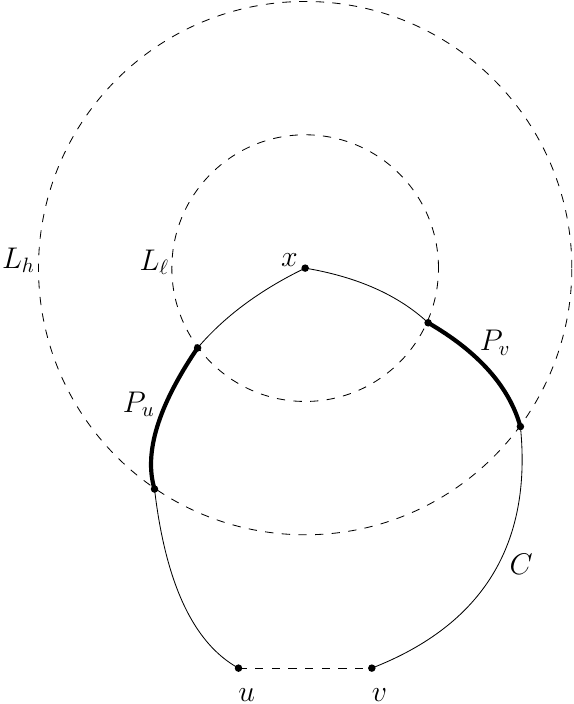}
 \caption{Proof of Lemma~\ref{lem-sep}. Solid curves consist of edges in $G$, whereas dashed curves may contain edges in the triangulation but not in $G$.}\label{fig:separator}
\end{figure}

\subsection{Ensuring simplicity of walks}\label{sec:simplicity}

In this subsection, we formulate a type of \emph{legal alternating} closed walks so that a shortest closed walk of this type is automatically a simple cycle.
This will be useful later in Section~\ref{sec:girth}.
(Warning: both the formulation and the proof of correctness below are somewhat subtle.)

In the following, we imagine that a separator has been computed, and the two parts
are colored red and blue, while the separator vertices are colored white.

\begin{lemma}\label{lem1}
Let $G$ be an undirected graph, where each vertex is colored red, blue, or white, and there are no edges between red and blue vertices.
Consider a closed walk $C$ which is the concatenation of simple paths $\pi^R_0,\tau_0,\pi^B_0,\tau'_0,\ldots,\pi^R_{k-1},\tau_{k-1},\pi^B_{k-1},\tau'_{k-1}$ with $k\ge 1$,
where each $\pi^R_i$ consists of red and white vertices only and has at least one red vertex,
each $\pi^B_i$ consists of blue and white vertices only and has at least one blue vertex,
and each $\tau_i$ and $\tau'_i$ consist of white vertices only (but may be empty).
Then $C$ must contain a nontrivial simple cycle.
\end{lemma}
\begin{proof}
Suppose that two consecutive edges of $C$ are identical, i.e., we see 3 consecutive vertices $u,v,u$ in $C$.
If $uv$ is the last edge of $\pi^R_i$ and $vu$ is the first edge of $\tau_i$, then
$u$ and $v$ are white vertices and we can remove $v$ from $C$ and still have a closed walk satisfying the stated condition.  
If $uv$ is the last edge of $\pi^R_i$, and $\tau_i$ is empty, and $vu$ is the first edge of $\pi^B_i$, then
$u$ and $v$ must again be white vertices and we can remove $v$ from $C$ and still have a closed walk satisfying the stated condition. 
After handling all other similar cases, we may assume that no two consecutive edges of $C$ are identical.
Now, $C$ forms a subgraph where every vertex has degree at least 2, and so must contain a nontrivial simple cycle.
\end{proof}

\begin{lemma}\label{lem-girth-alternate}
Let $G$ be an undirected graph, where each vertex is colored red, blue, or white, and there are no edges between red and blue vertices.
Let $R$, $B$, and $W$ be the subset of red, blue, and white vertices respectively.
Define the following:

\begin{itemize}
\item A simple path $\pi$ from $u$ to $v$ in $G$ is a \emph{legal red path}
if $\pi$ lies in $G[R\cup W]$, $u,v\in W$, and 
$\delta_{G[R\cup W]}(u,v)\neq \delta_{G[W]}(u,v)$.
\item A simple path $\pi$ from $u$ to $v$ in $G$ is a \emph{legal blue path}
if $\pi$ lies in $G[B\cup W]$, $u,v\in W$, and 
$\delta_{G[B\cup W]}(u,v)\neq \delta_{G[W]}(u,v)$.
\item A closed walk $C$ is a \emph{legal alternating walk} if it is the concatenation of simple paths 
$\pi^R_0,\tau_0,\pi^B_0,\tau'_0,$ $\ldots,\pi^R_{k-1},\tau_{k-1},\pi^B_{k-1},\tau'_{k-1}$ with $k\ge 1$,
where each $\pi^R_i$ is a legal red path, each $\pi^B_i$ is a legal blue path, and each $\tau_i$ and $\tau'_i$ lie in $G[W]$.
\end{itemize}

\noindent
Then
\begin{enumerate}
    \item[\rm (i)] the girth of $G$ is at most the length of a shortest legal alternating closed walk;
    \item[\rm (ii)] furthermore, equality holds if the girth of $G$ is not equal of girth of $G[R\cup W]$ or the girth of $G[B\cup W]$.
\end{enumerate}
\end{lemma}
\begin{proof}
For (i), consider a shortest legal alternating closed walk $C$.
Then $C$ is the concatenation of simple paths 
$\pi^R_0,\tau_0,\pi^B_0,\tau'_0,\ldots,\pi^R_{k-1},\tau_{k-1},\pi^B_{k-1},\tau'_{k-1}$  with $k\ge 1$,
where each $\pi^R_i$ is a legal red path, each $\pi^B_i$ is a legal blue path, and each $\tau_i$ and $\tau'_i$ lie in $G[W]$.
Each $\pi^R_i$ must pass through at least one red vertex, 
because $\delta_{G[R\cup W]}(u(\pi^R_i),v(\pi^R_i))\neq\delta_{G[W]}(u(\pi^R_i),v(\pi^R_i))$, where $u(\pi)$ and $v(\pi)$ denote the starting and ending vertices of a path~$\pi$.
Similarly, each $\pi^B_i$ must pass through at least one blue vertex.  So, the condition in Lemma~\ref{lem1} is satisfied, and 
$C$ must contain a nontrivial simple cycle.  Hence, the girth of $G$ is at most the length of $C$.

For (ii), consider a shortest simple cycle $C$ in $G$, of length $g$.
If there are multiple shortest simple cycles, pick one with the minimum number of non-white vertices.
We can express $C$ as a concatenation of simple paths $\pi^R_0,\pi^B_0,\ldots,\pi^R_{k-1},\pi^B_{k-1}$ with $k\ge 1$, 
where each $\pi^R_i$ lies in $G[R\cup W]$, uses at least one red vertex, and starts and ends with white vertices, and
each $\pi^B_i$ lies in $G[B\cup W]$, uses at least one blue vertex, and starts and ends with white vertices.
(Note that $C$ contains at least one red and at least one blue vertex, by the assumption in stated in~(ii).)
If $\delta_{G[R\cup W]}(u(\pi^R_i),v(\pi^R_i))=\delta_{G[W]}(u(\pi^R_i),v(\pi^R_i))$ for some $i$,
then we can replace $\pi^R_i$ with a shortest path $\gamma$ from $u(\pi^R_i)$ to $v(\pi^R_i)$ in $G[W]$, and
obtain a closed walk which is the concatenation of two non-identical simple paths $\gamma$ and $C\setminus \pi_i^R$, and so must contain a nontrivial simple cycle;
the cycle has length at most $g$ and uses fewer non-white vertices: a contradiction.
Thus, each $\pi^R_i$ is a legal red path.
Similarly, each $\pi^B_i$ is a legal blue path.
Hence, $C$ is a legal alternating closed walk.
\end{proof}

(An astute reader might wonder why we have overcomplicated the definition of legal alternating walk by including the white subpaths $\tau_i$, since part (ii) of the proof of Lemma~\ref{lem-girth-alternate} does not require them.  The purpose for including them is computational, as we will see later: without them, intermediate distance matrices may no longer satisfy the bounded-difference property.)

\newcommand{\AAA}{\widehat{A}}
\newcommand{\BBB}{\widehat{B}}
\newcommand{\CCC}{\widehat{C}}
\newcommand{\GGG}{{\widehat{G}}}
\newcommand{\HHH}{{\widehat{H}}}
\newcommand{\XXX}{\widehat{X}}
\newcommand{\PPP}{\widehat{P}}
\newcommand{\QQQ}{\widehat{Q}}
\newcommand{\SSS}{\widehat{S}}
\newcommand{\uuu}{\widehat{u}}
\newcommand{\vvv}{\widehat{v}}
\newcommand{\VVV}{\widehat{V}}

\section{Multiple Distances in Clustered Planar Graphs}\label{sec:dist}

As a warm-up, we first present an algorithm for a subproblem (of independent interest), on computing
all distances with sources and destinations from a given subset $X$, in an intersection graph of line segments.
It is actually more convenient to solve the problem for a more general family of graphs, which we call
\emph{clustered planar graphs}, where each vertex in a planar graph is replaced by a ``cluster'' of $c$ vertices:

\begin{definition}
A \emph{$c$-clustered planar graph} refers to the product graph $\GGG=G\times K_c$ where $G=(V,E)$ is a planar graph and $K_c$ is the $c$-clique.  In other words, the vertices of $\GGG$ are pairs $(v,j)\in V\times\{1,\ldots,c\}$, and there is an edge between two vertices $(u,i)$ and $(v,j)$ whenever $uv\in E$ or $u=v$.
\end{definition}

Known shortest-path-related algorithms for planar graphs do not necessarily generalize to clustered planar graphs (because of the lack of Monge property, VC-dimension arguments, etc.).  We prove the following result, which beats the obvious $\OO(xn)$ time bound, for example, when $x$ is around $\sqrt{n}$:

\newcommand{\nn}{\mu}
\newcommand{\bb}{\beta}

\begin{theorem}\label{thm:dist}
Let $\GGG=G\times K_c$ be a $c$-clustered planar graph, where $G=(V,E)$ has $n$ vertices and $c=O(1)$.
Suppose each edge of $\GGG$ has a nonnegative integer weight bounded by $O(1)$.
Given a subset $X$ of $x$ vertices in $G$, we can compute
$\delta_\GGG[X\times K_c,X\times K_c]$ in 
$\OO(n^{3/2-\alpha}+ n^{(1/2+\alpha)(3+\omega)/2} + x^{(3+\omega)/2} + x^2 n^{1/2-\alpha})$
expected time for any given $\alpha$.
\end{theorem}
\begin{proof}
We will solve an extended problem: given $G$, $\GGG$, $X$, and an additional subset $B$ of $\beta$ vertices in $G$ such that $B$ has $\gamma$ components in $G$, we want to compute $\delta_\GGG[(B\cup X)\times K_c,(B\cup X)\times K_c]$.  (Initially, $B=\emptyset$ and $\bb=\gamma=0$.)

\paragraph{Step 1: Build a separator.}
Give each vertex in $V\setminus B$ weight 1 and each vertex in $B$ weight $1+w$ for a parameter $w$ (assuming $n+w\bb\gg (1+w)/\eps$);
the total weight is $n+w\bb$.
Give each vertex in $V\setminus B$ load 1 and each vertex in $B$ load $n/\bb$; the total load is $O(n)$.
We apply Lemma~\ref{lem-sep} for a value $p$ to be set later, to
partition $V$ into $V_1,V_2,S$
such that 

\begin{itemize}
    \item $|V_i|+w|V_i\cap B|\le 3(n+w\bb)/4$ for each $i\in\{1,2\}$; 
    \item there are no edges of $G$ in $V_1\times V_2$;
    \item $S$ can be further partitioned into $\sepp$ and $\sepq$ such that $\sepp$ is the union of at most two paths in $G$ of length $O(p)$, and $Q$ has total load  at most $O(n/p)$---this implies that $|Q|\le O(n/p)$ and $|Q\cap B|\le O(\frac{n/p}{n/\bb})=O(\bb/p)$ simultaneously.
\end{itemize}

\noindent
Let $\VVV=V\times K_c$, $\SSS=S\times K_c$, $\PPP=P\times K_c$, $\QQQ=Q\times K_c$, $\BBB=B\times K_c$ and $\XXX=X\times K_c$.

\paragraph{Step 2: Recurse.}
For each $i\in\{1,2\}$, we recursively solve the problem for $G_i=G[V_i\cup S]$, $\GGG_i=G_i\times K_c$ (with edge weights inherited from $\GGG$), $B_i=(V_i\cap B)\cup P\cup (Q\cap B)$, and $X_i=V_i\cap X$.

Let $n_i = |V_i\cup S|$, $\beta_i=|B_i|$, 
$\gamma_i$ be the number of components of $B_i$ in $G_i$, and $x_i=|X_i|$.
Then $n_1+n_2\le n + O(p+n/p)$, 
$\beta_1+\beta_2\le \bb + O(p+\beta/p)$,
$n_i+w\bb_i\le 3(n+w\bb)/4 + O(w(p+n/p))$, 
and $x_1+x_2\le x$.  

To bound $\gamma_i$, consider a component $C$ of $B_i$ in $G_i$.  If $C$ does not visit $S$, then $C$  is an original component of $B$  (since $S\cap B\subseteq B_i$).  If $C$ visits a vertex of $P$, then $C$ contains one of the two paths in~$P$ (since $P\subseteq B_i$).  If $C$ visits a vertex of $Q\cap B$, then $C$ can be charged to that vertex.
It follows that $\gamma_1+\gamma_2\le\gamma+O(1+\bb/p)$.

\paragraph{Step 3: Compute \boldmath $\delta_\GGG[\BBB\cup\XXX,\BBB\cup\XXX]$.}
We now compute $\delta_\GGG[\BBB\cup\XXX,\BBB\cup\XXX]$ by the following expression:
\begin{eqnarray}
&& \delta_{\GGG_1}[\BBB\cup\XXX,\BBB\cup\XXX] \,\vee\, \delta_{\GGG_2}[\BBB\cup\XXX,\BBB\cup\XXX]\,\vee {}\label{expr0}\\
    && (\delta_\GGG[\BBB\cup\XXX,\QQQ]\star\delta_\GGG[\QQQ,\BBB\cup\XXX]) \,\vee {}\label{expr1}\\
    && \Big((\delta_{\GGG_1}[\BBB\cup\XXX,\PPP]\vee\delta_{\GGG_2}[\BBB\cup\XXX,\PPP])\star(\delta_{\GGG_1}[\PPP,\PPP]\star\delta_{\GGG_2}[\PPP,\PPP])^n \star {}\nonumber\\ &&\qquad\qquad\qquad\qquad  \qquad\qquad\qquad\qquad(\delta_{\GGG_1}[\PPP,\BBB\cup\XXX]\vee \delta_{\GGG_2}[\PPP,\BBB\cup\XXX])\Big),\label{expr2}
\end{eqnarray}  
where $M\vee M'$ denotes the entry-wise minimum of two matrices $M$ and $M'$, and $M^k$ denotes the $k$-th power under min-plus product.

The expression is correct because for $\uuu,\vvv\in \BBB\cup\XXX$, if a shortest path $\pi$ from $\uuu$ to $\vvv$ in $\GGG$ avoids $\SSS$, then its length is 
$\delta_{\GGG_1}(\uuu,\vvv)$ or $\delta_{\GGG_2}(\uuu,\vvv)$; otherwise, if $\pi$ passes through $\QQQ$, then its length is
$(\delta_\GGG[\BBB\cup\XXX,\QQQ]\star\delta_\GGG[\QQQ,\BBB\cup\XXX])(\uuu,\vvv)$;
otherwise, if $\pi$ passes through $\PPP$, starting from $\GGG_i$ and ending in $\GGG_j$, then its length is
$(\delta_{\GGG_1}[\BBB\cup\XXX,\PPP]\star(\delta_{\GGG_1}[\PPP,\PPP]\star\delta_{\GGG_2}[\PPP,\PPP])^n \star\delta_{\GGG_j}[\PPP,\BBB\cup\XXX])(\uuu,\vvv)$.

To evaluate the expression, we
first compute $\delta_{\GGG}[\VVV,\QQQ]$ and $\delta_{\GGG_i}[\VVV,\QQQ]$ for each $i\in\{1,2\}$  by running BFS from each vertex of $\QQQ$. Since $\abs{Q}=O(n/p)$, this part takes $\OO(n^2/p)$ time. 

The matrix $\delta_{\GGG_i}[\BBB\cup\XXX,\BBB\cup\XXX]$ is available, since $(B\cup X)\cap (V_i\cup S)\subseteq B_i\cup X_i\cup Q$, and we know $\delta_{\GGG_i}[\BBB_i\cup\XXX_i,\BBB_i\cup\XXX_i]$ from the two recursive calls, and we already know $\delta_{\GGG_i}[\VVV,\QQQ]$.

The matrices $\delta_{\GGG_i}[\BBB\cup\XXX,\PPP]$
and $\delta_{\GGG_i}[\PPP,\PPP]$ are also available, since $P\subseteq B_i$.

We now describe how to compute 
$\delta_\GGG[\BBB\cup\XXX,\QQQ]\star\delta_\GGG[\QQQ,\BBB\cup\XXX]$ in (\ref{expr1}).
For any component $\CCC\subseteq \BBB$, the columns of $\delta_\GGG[\CCC,\QQQ]$ satisfy the bounded-difference property, after reordering the rows according to an Euler tour of a spanning tree of $C$---some rows may need to be duplicated, but the number of rows increases by only a constant factor.  (Note that all vertices in a common cluster $\{v\}\times K_c$ are within a constant distance from each other.)  Since $B$ has $\gamma$ components, 
we can compute $\delta_\GGG[\BBB,\QQQ]\star\delta_\GGG[\QQQ,\BBB\cup\XXX]$ by Lemma~\ref{lem:bdd:diff} in $\OO(\gamma (\beta+n/p+x)^{(3+\omega)/2})$ time.
The product
$\delta_\GGG[\BBB\cup\XXX,\QQQ]\star\delta_\GGG[\QQQ,\BBB]$ 
is symmetric.
It remains to compute
$\delta_\GGG[\XXX,\QQQ]\star\delta_\GGG[\QQQ,\XXX]$: here, since we don't have the bounded-difference property, we naively compute the product in $O(|X|^2|Q|)=O(x^2n/p)$ time.

Finally, we describe how to compute (\ref{expr2}).
The columns of $\delta_{\GGG_i}[\PPP,\PPP]$ satisfy the bounded-difference property, after dividing the matrix into two parts and reordering the rows, since $P$ is the union of at most two walks; so, we can compute $\delta_{\GGG_1}[\PPP,\PPP]\star\delta_{\GGG_2}[\PPP,\PPP]$ in $\OO(p^{(3+\omega)/2})$ time.
We can compute the $n$-th power by repeated squaring using $O(\log n)$ bounded-difference matrix products. (Note that if the columns of a matrix $M$ satisfy the bounded-difference property, then the columns of a product $M\star M'$ also satisfy the bounded-difference property.)

The rows of $\delta_{\GGG_i}[\BBB\cup\XXX,\PPP]$ satisfy the bounded-difference
property (after dividing the matrix into two parts), and so the 
other products can be computed in a similar way.

\paragraph{Running time.}
The time complexity $T(n,\beta,\gamma,x)$ of the algorithm satisfies the following recurrence:

\[
T(n,\bb,\gamma,x)\le  \!\!\!\! \max_{\scriptsize\begin{array}{c}n_1,n_2,\bb_1,\bb_2,x_1,x_2:\\ n_i+w\bb_i\le 3(n+w\bb)/4+O(w(p+n/p)),\\n_1+n_2\le n+O(p+n/p),\\ \bb_1+\bb_2\le\bb+O(p+\beta/p),\\ \gamma_1+\gamma_2\le\gamma+O(1+\beta/p),\\
x_1+x_2\le x\end{array}}\!\!\!\!\!\!
\begin{array}[t]{l}\Big( T(n_1,\bb_1,\gamma_1,x_1)+ T(n_2,\bb_2,\gamma_2,x_2) + {}
\\[.5ex]\ \ \OO(n+n^2/p+{}\\ \ \  \ \ \gamma (\bb+p+n/p+x)^{(3+\omega)/2} +  x^2n/p)\Big).
\end{array}
\]

\paragraph{Solving the recurrence.}
Intuitively, one could choose $p\approx n^{1/2+\alpha}$, and think of $\beta$ as $\OO(p)$ and $\gamma$ as $\OO(1)$, and see that the cost is upper bounded by $n^{3/2-\alpha}+x^{(3+\omega)/2}+x^2n^{1/2-\alpha}$ for any sufficiently small positive $\alpha$ (as long as $\omega<3$), at least at the top few levels of the recursion.

Formally, let $r$ be a fixed parameter, to be set later, with $r\gg w(p+n/p)\log n$.
For a node $u$ in the recursion tree, let $n_u,\bb_u,\gamma_u,X_u$ denote its input parameters. 
Suppose we stop the recursion at a node $u$ whenever $n_u+w\bb_u\le r$; in particular, this condition implies $n_u\le r$ and $\bb_u\le r/w$ simultaneously. Then the parent $v$ of each leaf satisfies $\nn_v+w\bb_v>r$.  This recursion subtree has depth at most $O(\log\frac{n+w\bb}{r})$. 
Let $\xi$ be the size of this recursion subtree.
Then $\sum_{v} (n_v +w\bb_v) \le O((n + w\bb + \xi\cdot  w(p+n/p))\log\frac{n+w\bb}{r})$,
where the summation is over each $v$ that is the parent of some leaf. 
Thus,  $\xi\le  O(\frac{1}{r} \sum_v (n_v +w\bb_v)) \le 
O(\frac{n+w\beta}{r} \log\frac{n+w\beta}{r}) + o(\xi)$, implying
that $\xi\le O(\frac{n+w\beta}{r} \log\frac{n+w\beta}{r})$.

This leads to a new recurrence:
\[
T(n,\beta,\gamma,x) \le \max_{\scriptsize\begin{array}{l}x_1,x_2,\ldots:\\ \sum_k x_k\le x\end{array}}
\sum_{k=1}^{O(\frac{n+w\beta}{r} \log\frac{n+w\beta}{r})} 
\begin{array}[t]{l}
\big( T(r,r/w,\gamma+\OO(1+\bb/p),x_k) + {}\\[.5ex]
\ \ \OO\big(n + n^2/p+{}\\[.5ex]
\ \ \ \ \ \ \gamma (\bb+p+n/p+x)^{(3+\omega)/2} + x^2n/p)\big).
\end{array}
\]

We now set $\bb=n^{1/2+\alpha}$,  $w=n^{1/2-\alpha}$, $r=n/t$, and $p=n^{1/2+\alpha}/\log^{1+\eps} n$ for a sufficiently large constant~$t$.
Note that indeed $r\gg w(p+n/p)\log n$.
By loosely bounding $r/w$ by $(n/t)^{1/2+\alpha}$, the recurrence becomes:
\[
T(n,n^{1/2+\alpha},\gamma,x) \le \max_{\scriptsize\begin{array}{l}x_1,x_2,\ldots:\\ \sum_k x_k\le x\end{array}} 
\sum_{k=1}^{O(t\log t)} 
\begin{array}[t]{l}
\big( T(n/t,(n/t)^{1/2+\alpha},\gamma+\OO(1),x_k) + {}\\[.5ex]
\ \ \OO(n^{3/2-\alpha}+\gamma (n^{(1/2+\alpha)(3+\omega)/2}+ x^{(3+\omega)/2}) {}\\[.5ex]\ \ \ \ \ \ 
x^2n^{1/2-\alpha})\big).
\end{array}
\]
This solves to $T(n,n^{1/2+\alpha},\gamma,x) = \OO(n^{3/2-\alpha}+ (1+\gamma) (n^{(1/2+\alpha)(3+\omega)/2} + x^{(3+\omega)/2}) + x^2 n^{1/2-\alpha})$.
\end{proof}

\begin{corollary}\label{cor:segs:dist}
Given a set $A$ of line segments in $\R^2$ and a subset $X\subset A$ of $x$ segments, 
we can compute 
$\delta_G[X,X]$, where $G$ is the intersection graph of $A$, in 
$\OO(\nn^{3/2-\alpha}+ \nn^{(1/2+\alpha)(3+\omega)/2} + x^{(3+\omega)/2} + x^2 \nn^{1/2-\alpha})$
expected time for any given $\alpha$, where $\mu$ is the number of endpoints and intersection points of $A$. 
\end{corollary}
\begin{proof}
For simplicity, we assume that the given segments are in general position (e.g., no 3 segments intersect at a common point).

First we compute the intersections of the line segments in $A$, for example, by applying Bentley and Ottmann's plane-sweep algorithm~\cite{DBLP:journals/tc/BentleyO79},
which runs in $O(\nn\log\nn)$ time.

Let $H$ be the \emph{planarized intersection graph} of $A$, i.e., the planar graph formed by the arrangement of $A$,
where the vertices consist of the segment endpoints and intersection points of $A$, and the edges of $H$ correspond to pairs of adjacent vertices in the arrangement.  

Let $\HHH$ be the clustered planar graph $H\times K_4$ (see Figure~\ref{fig:clustered_planar_graph}).
Number the edges incident to each vertex of $H$
from $\{1,2,3,4\}$ in clockwise order.
Give the edge $(u,i)(v,j)$ weight 0 in $\HHH$ if $uv$ is in $H$ and the $i$-th edge incident to $u$ coincides with the $j$-th edge incident to $v$.  In addition, for
each vertex $u$ of degree~4 in $H$, give the edges $(u,1)(u,2)$, $(u,2)(u,3)$, $(u,3)(u,4)$, $(u,4)(u,1)$ weight 1, and give the edges $(u,1)(u,3)$ and $(u,2)(u,4)$ weight 0.
Give all other edges weight 2 in $\HHH$.
Let $u(s)$ denote the left endpoint of a segment $s$.
Then for any two segments $s$ and $t$,
the distance between $s$ and $t$ in $G$ is equal to the distance between $(u(s),1)$ and $(u(t),1)$ in $\HHH$.
So, the result immediately follows by applying Theorem~\ref{thm:dist} to $\HHH$ and $\{u(s): s\in X\}$.
\end{proof}

\begin{figure}
\centering
\includegraphics[scale=.5]{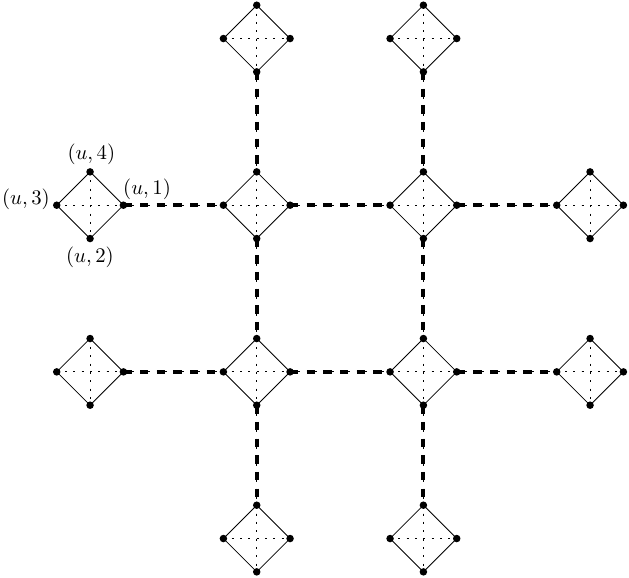}
 \caption{The clustered planar graph $\HHH$. The dotted edges and the dashed edges have weight 0, while the solid edges have weight 1. Edges with weight 2 are omitted.}\label{fig:clustered_planar_graph}
\end{figure}

\section{Girth for Line Segments}\label{sec:girth}

We now present our girth algorithm for intersection graphs of line segments.
It is tempting to continue to work within the clustered planar graph framework, but unfortunately a clustered planar graph by definition contains short cycles in every cluster (furthermore, when generalizing from segments to curves later, an additional difficulty is that simple cycles in the planarized intersection graph may not correspond to simple cycles in the intersection graph---see Figure~\ref{fig:curve}).  

We propose another recursive algorithm, which is also based on our separator lemma but proceeds differently: on one hand, we have to work with intersection graphs instead of clustered planar graphs, and also have to deal with the issue of ensuring simplicity of walks; on the other hand, Corollary~\ref{cor:segs:dist} is available as a subroutine to compute distances, so fewer input arguments are needed in the new recursive algorithm (notably, we don't need an additional input set $B$).

\begin{figure}
\centering
\includegraphics[scale=1]{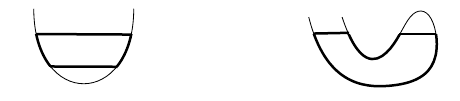}
 \caption{These are not simple cycles in the intersection graph!}\label{fig:curve}
\end{figure}

\begin{theorem}\label{thm:segs:girth}
Given a set $A$ of $n$ line segments in $\R^2$, 
we can compute the girth $g_G$ of its intersection graph $G$ of $A$
in $O(n^{1.483})$ expected time, under the promise that the girth is more than~4.
\end{theorem}
\begin{proof}
It is known that any intersection graph of $n$ line segments with girth more than~4 must have $O(n)$ edges~\cite{DBLP:journals/jgt/PachS08}.

\paragraph{Step 1: Build a separator.}
We first construct the planarized intersection graph $H=(V_H,E_H)$,
as in the proof of Corollary~\ref{cor:segs:dist},
in $\OO(n)$ time.  Note that $\mu=|V_H|=O(n)$.
Define the weight of a point $p\in V_H$ to be 1 if it is the left endpoint of a segment, and 0 otherwise; the total weight is $n$.
We can apply Lemma~\ref{lem-sep} (with unit loads), for a value $p$ to be set later, to
partition $V_H$ into $V_1,V_2,S_H$
such that 

\begin{itemize}
    \item $V_i$ has weight at most $3n/4$ for each $i\in\{1,2\}$;
    \item there are no edges of $H$ in $V_1\times V_2$;
    \item $\sep_H$ can be further partitioned into $\sepp_H$ and $\sepq_H$ such that $\sepp_H$ is the union of at most two paths in $H$ of length $O(p)$, and $|Q_H|= O(n/p)$.

\end{itemize}

\noindent
For each $i\in\{1,2\}$, let $A_i$ consist of the segments $s\in A$ for which all points of $V_H$ on $s$ are in $V_i$.
Note that $|A_i|\le 3n/4$.
Let $P$ consist of the segments with at least one point in $P_H$. Let $Q$ consist of the segments with at least one point in $Q_H$ and no points in $P_H$.
Let $S=P\cup Q$. Note that 
$P$ is the union of at most two walks of length $O(p)$ in $G$ (since any path in $H$ corresponds to a walk in $G$).

\paragraph{Step 2: Recurse.}
For each $i\in\{1,2\}$, we recursively solve the problem for $A_i\cup P$ to
obtain the girth $g_{G_i}$ of the intersection graph $G_i$ of $A_i\cup P$.  (Note that the assumption that $G$ has girth more than~4 implies that $G_i$ also has girth more than~4.)

Let $n_i=|A_i\cup P|$.  Then $n_1,n_2\le 3n/4 + O(p)$ and $n_1+n_2\le n + O(p)$.

\paragraph{Step 3: Compute \boldmath $\delta_{G_i}[P,P]$.}
For each $i\in\{1,2\}$, we compute $\delta_{G_i}[P,P]$ by Corollary~\ref{cor:segs:dist}
in $\OO(n^{3/2-\alpha}+ n^{(1/2+\alpha)(3+\omega)/2} + p^{(3+\omega)/2} + p^2 n^{1/2-\alpha})$
expected time.

\paragraph{Step 4: Compute the girth \boldmath $g_G$.}
We first run BFS from each segment of $Q$ to find the length $g_Q$ of the shortest cycle visiting $Q$ in $G$.  This part takes $\OO(n^2/p)$ time. 

\newcommand{\glegal}{g_{\textrm{legal}}}
Next, we compute the length $\glegal$ of the shortest legal alternating closed walk where the segments in $A_1$ are red, the segments in $A_2$ are blue, and the segments in $P$ are white (discarding the segments in $Q$).
By Lemma~\ref{lem-girth-alternate}, the girth $g_{G[A_1\cup A_2\cup P]}$ of $G[A_1\cup A_2\cup P]$ is equal to
$\min\{\glegal,g_{G_1},g_{G_2}\}$.  The final answer
is $g_G=\min\{g_{G[A_1\cup A_2\cup P]},g_Q \}$.

To compute $\glegal$,
for each $i\in\{1,2\}$, we first compute the ``legal'' distance matrices $\delta'_{G_i}[P,P]$: for $u,v\in P$, let
\[
\delta'_{G_i}(u,v)=\begin{cases}
    \delta_{G_i}(u,v) & \text{if } \delta_{G_i}(u,v)\ne \delta_{G[P]}(u,v),\\
    +\infty & \text{otherwise.}
\end{cases}
\]
We already know the matrix $\delta_{G_i}[P,P]$.
We can compute $\delta_{G[P]}$ naively in $\OO(p^\omega)$ time.
Consequently, we know the matrix $\delta'_{G_i}[P,P]$.

Next, we compute
\[ \delta''_{G_i}[P,P] = \delta'_{G_i}[P,P]\star\delta_{G[P]}.\]
Although $\delta'_{G_i}[P,P]$ may not satisfy a bounded-difference property,
the rows of $\delta_{G[P]}$ satisfy the bounded-difference property, after dividing the matrix into two parts and reordering the columns, since $P$ is the union of at most two walks.   So, we can compute $\delta''_{G_i}[P,P]$ in
$\OO(p^{(3+\omega)/2})$ time by Lemma~\ref{lem:bdd:diff}.
It also follows that the rows of $\delta''_{G_i}[P,P]$ satisfy the bounded-difference property.  
Finally, let 
\[ Z_k = \bigvee_{i=1}^k (\delta''_{G_1}[P,P]\star \delta''_{G_2}[P,P])^i. \]
From the definition of legal alternating walks, we see that
$Z_n(u,u)$ is precisely the length of the shortest legal alternating closed walk starting at $u\in P$.  Thus, $\glegal$ is the smallest diagonal entry in $Z_n$.  
We can compute $Z_n$ recursively by the following formula:
\[ Z_k =  (Z_{k/2}\star Z_{k/2}) \vee (\delta''_{G_1}[P,P]\star \delta''_{G_2}[P,P]). \]
(This assumes that $k$ is even, and the odd case is similar.)  This requires $O(\log n)$ min-plus matrix products, and all intermediate matrices satisfy the bounded-difference property.  So, $Z_n$ can be computed in $\OO(p^{(3+\omega)/2})$ time.

\paragraph{Running time.}
The time complexity $T(n)$ of the algorithm satisfies the following recurrence:
\[
T(n)\le \max_{\scriptsize\begin{array}{c}n_1,n_2:\\ n_1,n_2\le 3n/4+O(p),\\n_1+n_2\le n+O(p)\end{array}}
\begin{array}[t]{l}\Big( T(n_1)+ T(n_2) + \OO(n^2/p + {}
\\[.5ex]\qquad n^{3/2-\alpha}+ n^{(1/2+\alpha)(3+\omega)/2} + p^{(3+\omega)/2} + p^2 n^{1/2-\alpha})\Big).
\end{array}
\]
To balance the three terms $n^2/p$, $n^{(1/2+\alpha)(3+\omega)/2}$, and $p^2 n^{1/2-\alpha}$,
we set $p=n^{1/2+\rho}$ and $\alpha=3\rho$, 
with 
$\rho=\frac{3-\omega}{22+6\omega} > 0.0173$,
using the latest bound $\omega<2.371339$~\cite{DBLP:conf/soda/AlmanDWXXZ25}.  Then the recurrence becomes:
\[
T(n)\le \max_{\scriptsize\begin{array}{c}n_1,n_2:\ n_1,n_2\le 3n/4+o(n),\\n_1+n_2\le n+o(n)\end{array}}
\big( T(n_1)+ T(n_2) + \OO(n^{3/2-\rho})\big),
\]
which solves to $T(n)=\OO(n^{3/2-\rho})\le O(n^{1.483})$.
\end{proof}

Finally, we remove the assumption that girth is more than 4:

\begin{corollary}\label{cor:segs:girth}
Given $n$ line segments in $\R^2$, 
we can compute the girth of the intersection graph
in $O(n^{1.483})$ expected time. 
\end{corollary}

\begin{proof}
We can check whether the intersection graph $G$ contains a 3-cycle
in $O(n^{1.408})$ time by an algorithm of Chan~\cite[Theorem 5.1]{DBLP:conf/soda/Chan23}, or 
a 4-cycle in $\OO(n)$ time by another algorithm of Chan~\cite[Theorem 6.2]{DBLP:conf/soda/Chan23}.
So, we may assume that the girth is greater than~4, and
the result follows from Theorem~\ref{thm:segs:girth}.
\end{proof}

\section{Extension to Algebraic Curves and Semialgebraic Sets}

In this section, we comment on how to extend our result more generally to algebraic curves or semialgebraic sets, assuming that each object is connected and has constant description complexity.

First, Corollary~\ref{cor:segs:dist} and Theorem~\ref{thm:segs:girth} immediately generalize to algebraic curves of constant complexity 
(string graphs without 4-cycles also have linear size~\cite{DBLP:journals/cpc/FoxP10},
and since each pair of curves with constant description complexity intersects $O(1)$ times, the planarized intersection graph has linear size as well).
Thus, the only nontrivial part is the algorithms for girth 3 and 4 by 
Chan~\cite[Theorems 5.1 and 6.2]{DBLP:conf/soda/Chan23}, which we have invoked in the proof of Corollary~\ref{cor:segs:girth}.
The main ingredient in these algorithms is the construction of a biclique cover of near-linear size, when the intersection graph is triangle-free.  Namely, we need to extend~\cite[Lemma 4.5]{DBLP:conf/soda/Chan23} for line segments to the more general case of algebraic curves, as formulated below:

\begin{lemma}\label{lem:bicliques:curves}
Consider the intersection graph $G=(V,E)$ of $n$ algebraic curves in $\R^2$ with constant description complexity. We can either find a triangle in $G$, or compute a \emph{biclique cover}, i.e., pairs of vertex subsets  $\{(A_1,B_1),\dots,(A_s,B_s)\}$ such that the union of $A_i\times B_i$ over all $i$ gives us precisely the edge set $E$, and $\sum_i (|A_i|+|B_i|)=\OO(n)$.  The construction time is $\OO(n)$. Furthermore, in this cover, each element appears in $\OO(1)$ subsets $A_i$ and $\OO(1)$ subsets $B_i$.
\end{lemma}

To prove the lemma, we first cut each curve into $O(1)$ $x$-monotone pieces.  (Cutting the curves might change the girth, but is fine for the purposes of constructing a biclique cover.)
We can then follow the proof of \cite[Lemma 4.5]{DBLP:conf/soda/Chan23}, using segment trees.  Every step works for $x$-monotone curves except for the following subproblem:

\begin{fact}\label{fact1}
Consider a vertical slab $\sigma$ and $n$ ``long'' $x$-monotone curves with left endpoints on the left wall of $\sigma$ and
right endpoints on the right wall of $\sigma$.  We can either find a 3-cycle in their intersection graph, or compute a 2-coloring in $O(n)$ time, after sorting the left endpoints.
\end{fact}

For line segments, this fact is easy to show by a simple incremental/greedy algorithm.
For general $x$-monotone curves that may pairwise intersect more than once, there is also an incremental algorithm but it is slightly trickier.  One way to view this subproblem is via Dilworth theorem: the ``aboveness'' relation defines a partial order, and triangle-freeness implies that the maximum antichain size of the poset is 2, and so the poset can be covered by 2 chains.
An $O(n)$-time algorithm is explicitly described, for example, in \cite[Section~4.1]{DBLP:journals/order/FelsnerRS03}, which outputs either an antichain of size~3 or a cover by 2 chains, for any poset where a linear extension is given (in our case, a sorted ordering of the left endpoints is given).
The rest of Chan's triangle-detection algorithm requires no modification.
Assuming that the intersection graph is triangle-free, Chan's 4-cycle detection algorithm, when given the biclique cover, generalizes as well without modification.

Lastly, we note how to generalize the result for connected semialgebraic sets.
The key observation is that if the girth is more than 3, then in any cycle $s_1,\ldots,s_g,s_1$, there can be no containment pairs: if $s_i$ is contained in $s_j$, then $s_{i-1},s_i,s_j$ would form a triangle.
Thus, we can replace each semialgebraic set $s$ with their boundary $\partial s$, which is a connected semialgebraic curve if there is no hole.  If $s$ has holes but is connected, we can attach curves to $\partial s$ to obtain a single connected curve that is inside $s$ and includes $\partial s$.  We can then solve the problem for the resulting curves.

The remaining case is girth 3.  The triangle-detection algorithm for algebraic curves
extends to semialgebraic sets easily.  For example, Fact~\ref{fact1} immediately extends, by the same connection to Dilworth's theorem, and consequently so does Lemma~\ref{lem:bicliques:curves}.

\begin{corollary}\label{cor:girth:curves}
Given $n$ algebraic curves or semialgebraic sets in $\R^2$ where each object is connected and has constant description complexity, 
we can compute the girth of the intersection graph
in $O(n^{1.483})$ expected time. 
\end{corollary}

Our results are unlikely to be optimal.  
Although we do not know how to improve Corollary~\ref{cor:segs:dist},
we are currently working on ideas that might further improve 
the time bound of Corollary~\ref{cor:segs:girth} for the case of line segments, by exploiting special properties about intersection graphs with large girth.

\bibliography{b}
\bibliographystyle{plainurl}

\end{document}